\documentclass[letterpaper,11pt]{article}
\usepackage{url,hyperref}
\usepackage{amsmath, amssymb}
\usepackage{color}
\usepackage{fullpage}

\usepackage[noend]{algorithmic}
\usepackage{algorithm}

\usepackage{enumerate}
\usepackage{fixltx2e}
\usepackage{hyperref}
\usepackage{times}
\usepackage{graphicx}
\usepackage{caption}
\usepackage{subcaption}
\usepackage{zerosum,csquotes}
\usepackage{enumitem}

 \newtheorem{theorem}{Theorem}[section]
 \newtheorem{lemma}[theorem]{Lemma}

 \newtheorem{definition}[theorem]{Definition}

\makeatletter
\def\GrabProofArgument[#1]{ #1: \egroup\ignorespaces}
\def\proof{\noindent\textbf\bgroup Proof%
	\@ifnextchar[{\GrabProofArgument}{. \egroup\ignorespaces}}

\makeatother

\newcommand{\opt}{\text{opt}}

\newcommand*\samethanks[1][\value{footnote}]{\footnotemark[#1]}
\newcommand{\wopt}{w_{\opt}}

\newcommand{\eat}[1]{}

\newcommand{\pcopt}{\text{opt\textsubscript{IP}}}
\newcommand{\lpopt}{\text{opt\textsubscript{SF}}}

\DeclareMathOperator*{\argmin}{arg\,min}

\newcommand{\dbsf}{\text{\sc DB-SF}}
\newcommand{\dbst}{\text{\sc DB-ST}}
\newcommand{\ewdbsf}{\text{\sc EW-DB-SF}}
\newcommand{\ewdbst}{\text{\sc EW-DB-ST}}

\newcommand{\odbsf}{\text{\sc online degree-bounded Steiner forest}}
\newcommand{\oewdbsf}{\text{\sc online edge-weighted degree-bounded Steiner forest}}

\DeclareMathOperator{\CC}{CC}

\newcommand{\x}{\mathbf{x}}

\newcommand{\Mc}[1]{\mathcal{#1}}
	\newenvironment{proofof}[1]{{\bf Proof of #1:  }}{\hfill\rule{2mm}{2mm}}
	
\newcounter{proccnt}

\title{Online Weighted Degree-Bounded Steiner Networks\\via Novel Online Mixed Packing/Covering}
\author{
Sina Dehghani \thanks{University of Maryland. email: \texttt{\{dehghani,ehsani,seddighin\}@umd.edu hajiagha@cs.umd.edu} }
\thanks{Supported in part by NSF CAREER award 1053605, NSF grant CCF-1161626, ONR YIP award N000141110662, DARPA/AFOSR grant FA9550-12-1-0423, and a Google faculty research award.}
\and Soheil Ehsani \samethanks[1] \samethanks[2]
\and MohammadTaghi Hajiaghayi \samethanks[1] \samethanks[2]
\and Vahid Liaghat \thanks{Stanford University. email: vliaghat@stanford.edu}
\and Harald R\"{a}cke  \thanks{Technische Universit\"{a}t M\"{u}nchen. email: raecke@in.tum.de}
\and Saeed Seddighin \samethanks[1] \samethanks[2]
}

\begin{document}

\sloppy

%
%

\date{}
\maketitle
\thispagestyle{empty}

\begin{abstract}
	We design the first online algorithm with poly-logarithmic competitive ratio for the \textit{edge-weighted degree-bounded Steiner forest} (\ewdbsf{}) problem and its generalized variant. We obtain our result by demonstrating a new generic approach for solving mixed packing/covering integer programs in the online paradigm. In \ewdbsf{}, we are given an edge-weighted graph with a degree bound for every vertex. Given a root vertex in advance, we receive a sequence of terminal vertices in an online manner. Upon the arrival of a terminal, we need to augment our solution subgraph to connect the new terminal to the root. The goal is to minimize the total weight of the solution while respecting the degree bounds on the vertices. In the offline setting, \textit{edge-weighted degree-bounded Steiner tree} (\ewdbst{}) and its many variations have been extensively studied since early eighties. Unfortunately, the recent advancements in the online network design problems are inherently difficult to adapt for degree-bounded problems. In particular, it is not known whether the fractional solution obtained by standard primal-dual techniques for mixed packing/covering LPs can be rounded online. In contrast, in this paper we obtain our result by using structural properties of the optimal solution, and reducing the \ewdbsf{} problem to an exponential-size mixed packing/covering integer program in which every variable appears only once in covering constraints. We then design a generic \textit{integral} algorithm for solving this restricted family of IPs.

As mentioned above, we demonstrate a new technique for solving mixed packing/covering integer programs. Define the \textit{covering frequency} $k$ of a program as the maximum number of covering constraints in which a variable can participate. Let $m$ denote the number of packing constraints. We design an online \textit{deterministic integral algorithm} with competitive ratio of $O(k\log m)$ for the mixed packing/covering integer programs.
We prove the tightness of our result by providing a matching lower bound for any randomized algorithm. We note that our solution solely depends on $m$ and $k$. Indeed, there can be exponentially many variables. Furthermore, our algorithm directly provides an integral solution, even if the integrality gap of the program is unbounded. We believe this technique can be used as an interesting alternative for the standard primal-dual techniques in solving online problems.

\end{abstract}

\newpage
\section{Introduction}
\textit{Degree-bounded} network design problems comprise an important family of
network design problems since the eighties. Aside from various real-world
applications such as vehicle routing and communication
networks~\cite{bauer1995degree,oliveira2005survey,voss1992problems}, the family
of degree-bounded problems has been a testbed for developing new ideas and
techniques.
The problem of \textit{degree-bounded spanning tree}, introduced in Garey and
Johnson's \textit{Black Book} of NP-Completeness~\cite{GJ79}, was first
investigated in the pioneering work of F{\"u}rer and Raghavachari~\cite{FR90}
(Allerton'90). In this problem, we are required to find a spanning tree of a
given graph with the goal of minimizing the maximum degree of the vertices in
the tree. Let $b^*$ denote the maximum degree in the optimal spanning tree.
F{\"u}rer and Raghavachari give a parallel approximation algorithm which
produces a spanning tree of degree at most $O(\log(n) b^*)$. This result was
later generalized by Agrawal, Klein, and Ravi~\cite{AKR91a} to the case of
degree-bounded Steiner tree (\dbst{}) and degree bounded Steiner forest
(\dbsf{}) problem. In \dbst{}, given a set of terminal vertices, we need to
find a subgraph of minimum maximum degree that connects the terminals. In the
more generalized \dbsf{} problem, we are given pairs of terminals and the
output subgraph should contain a path connecting each pair. F{\"u}rer and
Raghavachari~\cite{FR94}(SODA'92, J.~of Algorithms'94) significantly improved
the result for \dbsf{} by presenting an algorithm which produces a Steiner
forest with maximum degree at most $b^*+1$.

The study of \dbst{} and \dbsf{} was the starting point of a very popular line
of work on various degree-bounded network design
problems; e.g. \cite{MRSR98,N12,LS13,KKN13,EV14} and more recently~\cite{fukunaga2012iterative,ene2014improved}. One particular variant that has
been extensively studied was initiated by Marathe~\textit{et al.}~\cite{MRSR98}
(J.~of Algorithms'98): In the \textit{edge-weighted degree-bounded spanning
  tree} problem, given a weight function over the edges and a degree bound $b$,
the goal is to find a minimum-weight spanning tree with maximum degree at most
$b$. The initial results for the problem generated much interest in obtaining
approximation algorithms for the edge-weighted degree-bounded spanning tree
problem~\cite{chaudhuri2009would, chaudhuri2006push, goemans2006minimum,
  klein2004approximation, konemann2000matter, konemann2003primal,
  lau2009survivable, raghavachari1996algorithms, ravi2001approximation,
  ravi2006delegate}. The groundbreaking results obtained by Goemans~\cite{G06}
(FOCS'06) and Singh and Lau~\cite{LS07} (STOC'07) settle the problem by giving
an algorithm that computes a minimum-weight spanning tree with degree at most
$b+1$. Singh and Lau~\cite{LS13} (STOC'08) generalize their result for the
\textit{edge-weighted Steiner tree} (\ewdbst{}) and \textit{edge-weighted
  Steiner forest} (\ewdbsf{}) variants. They design an algorithm that finds a
Steiner forest with cost at most twice the cost of the optimal solution while
violating the degree constraints by at most three.

Despite these achievements in the offline setting, it was not known whether
degree-bounded problems are tractable in the \textit{online setting}. The
online counterparts of the aforementioned Steiner problems can be defined as
follows. The underlying graph and degree bounds are known in advance. The
demands arrive one by one in an online manner. At the arrival of a demand, we
need to augment the solution subgraph such that the new demand is satisfied.
The goal is to be competitive against an offline optimum that knows the demands
in advance.

Recently, Dehghani et al.~\cite{Dehghani2016} (SODA'16)
explore the tractability of the Online \dbsf{} problem by showing that a
natural greedy algorithm produces a solution in which the degree bounds are
violated by at most a factor of $O(\log n)$, which is asymptotically
\textit{tight}. They analyze their algorithm using a dual fitting approach
based on the combinatorial structures of the graph such as the
toughness\footnote{The toughness of a graph is defined as $\min_{X\subseteq V}
  \frac{|X|}{|\CC(G\setminus X)|}$; where for a graph $H$, $\CC(H)$ denotes the
  collection of connected components of $H$.} factor. Unfortunately, greedy
methods are not competitive for the edge-weighted variant of the problem.
Hence, it seems unlikely that the approach of \cite{Dehghani2016} can be
generalized to \ewdbsf{}.

The \textit{online edge-weighted Steiner connectivity} problems (with no bound
on the degrees) have been extensively studied in the last decades. Imase and
Waxman~\cite{IW91} (SIAM J. D. M.'91) use a dual-fitting argument to show that
the greedy algorithm has a competitive ratio of $O(\log n)$, which is also
asymptotically tight. Later the result was generalized to the EW SF variant by
Awerbuch \textit{et al.}~\cite{AAB96} (SODA'96) and Berman and
Coulston~\cite{BC97} (STOC'97). In the past few years, various primal-dual
techniques have been developed to solve the more general node-weighted
variants~\cite{AAABN09,NPS11,HLP13} (SIAM'09, FOCS'11, FOCS'13), prize-collecting
variants~\cite{QW11,HLP14} (ICALP'11,ICALP'14), and multicommodity
buy-at-bulk~\cite{chakrabarty2015online} (FOCS'15). These results are obtained by developing various primal-dual techniques~\cite{AAABN09,HLP13} while generalizing the application of combinatorial properties to the online setting~\cite{NPS11,HLP14,chakrabarty2015online}. In this paper however, we develop a primal approach for solving \textit{bounded-frequency mixed packing/covering integer programs}. We believe this framework would be proven useful in attacking other online packing and covering problems.



\subsection{Our Results and Techniques}
In this paper, we consider the online Steiner tree and Steiner forest problems at the presence of
both edge weights and degree bounds. In the Online \ewdbsf{} problem, we are
given a graph $G=(V,E)$ with $n$ vertices, edge-weight function $w$, degree
bound $b_v$ for every $v\in V$, and an online sequence of connectivity demands
$(s_i,t_i)$. Let $w_{\opt}$ denote the minimum weight subgraph which satisfies
the degree bounds and connects all demands. Let $\rho=\frac{\max_e
  w(e)}{\min_{e: w(e)>0} w(e)}$\footnote{Our competitive ratios have a 
  logarithmic dependency on $\rho$, i.e., the ratio between largest and
  smallest weight. It follows from the result of~\cite{Dehghani2016} that one
  cannot obtain polylogarithmic guarantees if this ratio is not polynomially bounded}.

\begin{theorem}\label{thm:maindbsf}
  There exists an online deterministic algorithm which finds a subgraph with
  total weight at most $O(\log^2 n) w_{\opt}$ while the degree bound of a
  vertex is violated by at most a factor of $O(\log^2(n) \log(n\rho))$.
	
	If one favors the degree bounds over total weight, one can find a subgraph
    with degree-bound violation
    $O(\log^2(n)\frac{\log(n\rho))}{\log\log(n\rho)})$ and total cost
    $O(\log^2(n)\frac{\log(n\rho))}{\log\log(n\rho)}) w_{\opt}$.
\end{theorem}

Our technical contribution for solving the \ewdbsf{} problem is twofold.
First by exploiting a structural result and massaging the optimal solution, we
show a formulation of the problem that falls in the restricted family of
\textit{bounded-frequency mixed packing/cover IPs}, while losing only
logarithmic factors in the competitive ratio. We then design a generic online 
algorithm with a logarithmic competitive ratio that can solve any instance of
the bounded-frequency packing/covering IPs. In what follows, we describe these
contributions in detail.

\paragraph{Massaging the optimal solution}
Initiated by work of Alon \textit{et al.}~\cite{AAABN09} on online set cover, Buchbinder and Naor
developed a strong framework for solving packing/covering LPs \textit{fractionally}
online. For the applications of their general framework in solving numerous
online problems, we refer the reader to the survey in \cite{BN09}. Azar
\textit{et al.}~\cite{ABFP13} generalize this method for the fractional
\textit{mixed} packing and covering LPs.
The natural linear program relaxation for \ewdbsf{}, commonly used
in the literature, is a special case of mixed packing/covering LPs: one needs
to select an edge from every cut that separates the endpoints of a demand
(covering constraints), while for a vertex we cannot choose more than a
specific number of its adjacent edges (packing constraints). Indeed, one can
use the result of Azar \textit{et al.}~\cite{ABFP13} to find an online
\textit{fractional} solution with polylogarithmic competitive ratio. 
However, doing the rounding in an online manner seems very hard.

Offline techniques for solving degree-bounded problems often fall in
the category of iterative and dependent rounding methods. Unfortunately,
these methods are inherently difficult to adapt for an online settings since
the underlying fractional solution may change dramatically in between the
rounding steps. Indeed, this might be the very reason that despite many
advances in the online network design paradigm in the past two decades, the
natural family of degree-bounded problems has remained widely open. In this
paper, we circumvent this by reducing \ewdbst{} to a novel formulation beyond
the scope of standard online packing/covering techniques and solving it using a
new online integral approach.

The crux of our IP formulation is the following structural property: Let
$(s_i,t_i)$ denote the $i^{th}$ demand. We need to augment the solution
$Q_{i-1}$ of previous steps by buying a subgraph that makes $s_i$ and $t_i$
connected. Let $G_i$ denote the graph obtained by contracting the pairs of
vertices $s_j$ and $t_j$ for every $j<i$. Note that any $(s_i-t_i)$-path in
$G_i$ corresponds to a feasible augmentation for $Q_{i-1}$. Some edges in $G_i$
might be already in $Q_{i-1}$ and therefore by using them again we can save
both on the total weight and the vertex degrees. However, in Section~\ref{sec:rightLP} we
prove that there always exists a path in $G_i$ such that even without sharing on any
of the edges in $G_i$ and therefore paying completely for the increase in the
weight and degrees, we can approximate the optimal solution up to a logarithmic
factor. This in fact, enables us to have a formulation in which the covering
constraints for different demands are \textit{disentangled}. Indeed, we 
only have one covering constraint for each demand. Unfortunately, this implies
that we have exponentially many variables, one for each possible
path in $G_i$. This may look hopeless since the competitive factors obtained by
standard fractional packing/covering methods introduced by Buchbinder and
Naor~\cite{BN09} and Azar \textit{et al.}~\cite{ABFP13}, depend on the
logarithm of the number of variables. Therefore we come up with a new approach
for solving this class of mixed packing/covering integer programs (IP).

\paragraph{Bounded-frequency mixed packing/covering IPs}
We derive our result for \ewdbst{} by demonstrating a new technique for solving mixed packing/covering \textit{integer} programs. We believe this approach could be applicable to a broader range of online problems.
The integer program~\ref{IP:pc} describes a general mixed packing/covering IP with the set of integer variables $\x\in \mathbb{Z}_{\geq 0}^n$ and $\alpha$. The packing constraints are described by a $m\times n$ non-negative matrix $P$. Similarly, the $q\times n$ matrix $C$ describes the covering constraints. The \textit{covering frequency} of a variable $x_i$ is defined as the number of covering constraints in which $x_i$ has a positive coefficient. The covering frequency of a mixed packing/covering program is defined as the maximum covering frequency of its variables.

\begin{align*}
\text{minimize}           &\quad \quad \alpha \tag{{$\mathbb{IP}1$}}\label{IP:pc}\enspace ,\\
s.t. &\quad\quad  P\x \leq \alpha\enspace .\\
&\quad\quad C\x \geq 1 \enspace .\\
&\quad\quad \x \in \mathbb{Z}_{\geq 0}, \alpha\in \mathbb{R}_{> 0}\enspace .
\end{align*}

In the online variant of mixed packing and covering IP, we are given the packing constraints in advance. However the covering constraints arrive in an online manner. At the arrival of each covering constraint, we should \textit{increase} the solution $\x$ such that it satisfies the new covering constraint. We provide a novel algorithm for solving online mixed packing/covering IPs.

\begin{theorem}\label{thm:mixedIP}
  Given an instance of the online mixed packing/covering IP, there exists a
  \textit{deterministic integral} algorithm with competitive $O(k\log
  m)$, where $m$ is the number of packing constraints and $k$ is the covering
  frequency of the IP.
\end{theorem}
We note that the competitive ratio of our algorithm is independent of the number of variables or the number of covering constraints. Indeed, there can be exponentially many variables.

Our result can be thought of as a generalization of the work of Aspnes et al.~\cite{aspnes1997line} (JACM'97) on virtual circuit routing. Although not explicit, their result can be massaged to solve mixed packing/covering IPs in which all the coefficients are zero or one, and the covering frequency is one. They show that such IPs admit a $O(\log(m))$-competitive algorithms. Theorem~\ref{thm:mixedIP} generalizes their result to the case with arbitrary non-negative coefficients and any bounded covering frequency.

We complement our result by proving a matching lower bound for the competitive ratio of any \textit{randomized} algorithm. This lower bound holds even if the algorithm is allowed to return fractional solutions.

\begin{theorem}\label{thm:mixedIPlower}
	Any randomized online algorithm $A$ for integral mixed packing and covering is
	$\Omega(k\log m)$-competitive, where $m$ denotes the number of packing
	constraints, and $k$ denotes the covering frequency of the IP. This even holds if $A$ is allowed to return a fractional solution.
\end{theorem}

As mentioned before, Azar \textit{et al.}~\cite{ABFP13} provide a fractional algorithm for mixed packing/covering LPs with competitive ratio of $O(\log m \log d)$ where $d$ is the maximum number of variables in a single constraint.
They show an almost matching lower bound for deterministic algorithms.
We distinguish two advantages of our approach compared to that of Azar \textit{et al}:

\begin{itemize}\itemsep=-0.05cm
	\item The algorithm in \cite{ABFP13} outputs a \textit{fractional} competitive solution which then needs to be rounded online. For various problems such as Steiner connectivity problems, rounding a solution online is very challenging, even if offline rounding techniques are known. Moreover, the situation becomes hopeless if the integrality gap is unbounded. However, for bounded-frequency IPs, our algorithm directly produces an integral competitive solution even if the integrality gap is large.
	\item Azar \textit{et al.} find the best competitive ratio with respect to the number of packing constraints and the size of constraints. Although these parameters are shown to be bounded in several problems, in many problems such as connectivity problems and flow problems, formulations with exponentially many variables are very natural. Our techniques provide an alternative solution with a tight competitive ratio, for formulations with bounded covering frequency.
\end{itemize}


\subsection{Preliminaries}
Let $G=(V, E)$ be an undirected graph of size $n$ ($|V|=n$). Let $w:E
\rightarrow \mathbb{Z}_{>0}$ be a function denoting the edge weights. For a
subgraph $H \subseteq G$, we define $w(H):=\sum_{e \in E(H)}w(e)$. For every
vertex $v \in V$, let $b_v \in \mathbb{Z}_{>0}$ denote the degree bound of $v$.
Let $\deg_H(v)$ denote the degree of vertex $v$ in subgraph $H$. We define the
load $l_H(v)$ of vertex $v$ w.r.t. $H$ as $\deg_H(v)/b_v$. In \dbsf{} we are given graph
$G$, degree bounds, and $k$ connectivity demands. Let $\sigma_i$ denote the
$i$-th demand. The $i$-th demand is a pair of vertices $\sigma_i=(s_i, t_i)$,
where $s_i, t_i \in V$. In \dbsf{} the goal is to find a subgraph $H \subseteq
G$ such that for each demand $\sigma_i$, $s_i$ is connected to $t_i$ in $H$,
for every vertex $v \in V$, $l_H(v) \leq 1$, and $w(H)$ is minimized. In this
paper without loss of generality we assume the demand endpoints are distinct
vertices with degree one in $G$ and degree bound infinity.

In the online variant of the problem, we are given graph $G$ and degree bounds in
advance. However the sequence of demands are given one by one. At arrival of
demand $\sigma_i$, we are asked to provide a subgraph $H_i$, such that $H_{i-1}
\subseteq H_i$ and $s_i$ is connected to $t_i$ in $H_i$.

The following integer program is a natural mixed packing and covering integer program for \oewdbsf{}. Let $\Mc{S}$ denote the collection of subsets of vertices that separate the endpoints of at least one demand. For a set of vertices $S$, let $\delta(S)$ denote the set of edges with exactly one endpoint in $S$. In \ref{IP:sf}, for an edge $e$, $x_e=1$ indicates that we include $e$ in the solution while $x_e=0$ indicates otherwise. The variable $\alpha$ indicates an upper bound on the violation of the load of every vertex and an upper bound on the violation of the weight. The first set of constraints ensures that the load of a vertex is upper bounded by $\alpha$. The second constraint ensures that the violation for the weight is upper bounded by $\alpha$. The third set of constraints ensures that the endpoints of every demand are connected. 
	\begin{align}
	\text{minimize}           &\quad \alpha\enspace. \tag{{$\mathbb{SF\_IP}$}}\label{IP:sf}\\
	\forall v\in V            &\quad \frac{1}{b_v}\sum_{e \in \delta(\smash{\{v\}})}{x_e} \leq \alpha \enspace.\label{IP:sfc1}\\
	&\quad \frac{1}{w_{\text{opt}}}\sum_{e\in E}{w(e)x_e} \leq \alpha \enspace.\label{IP:sfc2}\\
	\forall S\subseteq \Mc{S} &\quad  \sum_{e \in \delta\smash{(S)}}{x_e} \geq 1\enspace.\label{IP:sfc3}\\
	&\quad x_e \in \{0, 1\}, \alpha\in \mathbb{Z}_{> 0}\enspace.\nonumber
	\end{align}

\subsection{Overview of the Paper}
We begin Section \ref{sec:rightLP} by providing an online bounded frequency mixed packing/covering IP for \ewdbsf{}. Further we prove that this formulation has plausible structures. In Section \ref{sec:OMPC} we provide an online deterministic algorithm for online bounded frequency mixed packing/covering IPs. In Section \ref{sec:everything} we merge the IP formulation in Section \ref{sec:rightLP} and the techniques in Section \ref{sec:OMPC} to obtain online polylogarithmic-competitive algorithms for \ewdbsf{}. Finally in Section \ref{sec:lower} we complement our algorithm for online bounded frequency mixed packing/covering IPs by providing a matching lower bound for the competitive ratio of any randomized algorithm.
\section{Finding the Right Integer Program} \label{sec:rightLP}
	In this section we design an online mixed packing and covering integer program for\\ \oewdbsf. We show this formulation is near optimal, i.e. any $f-$approximation for this formulation, implies an $O(f\log^2 n)$-approximation for \oewdbsf{}. In Section \ref{sec:everything} we show there exists an online algorithm that finds an $O(\log^3 n)$-approximation of $w_{\text{opt}}$ and violates degree bounds by $O(\log^3 n \log w_{\text{opt}})$, where $w_{\text{opt}}$ denotes the optimal weight.
	
	First we define some notations. 
	For a sequence of demands $\sigma=\langle (s_1, t_1), \ldots, (s_k, t_k) \rangle$, we define $R_\sigma(i)$ to be a set of $i$ edges, connecting the endpoints of the first $i$ demands. In particular $R_\sigma(i):=\bigcup_{j=1}^i{e(s_j, t_j)}$, where $e(s_j, t_j)$ denotes a direct edge from $s_j$ to $t_j$. Moreover, we say subgraph $H_i$ satisfies the connectivity of demand $\sigma_i=(s_i, t_i)$, if $s_i$ and $t_i$ are connected in graph $H_i \cup R_\sigma(i-1)$. Let $\mathcal{H}_i$ denote the set of all subgraphs that satisfy the connectivity of demand $\sigma_i$. In \ref{IP:good} variable $\alpha$ denotes the violation in the packing constraints. Furthermore for every subgraph $H \subseteq G$ and demand $\sigma_i$, there exists a variable $x_H^i \in \{0, 1\}$. $x_H^i=1$ indicates we add the edges of $H$ to the existing solution, at arrival of demand $\sigma_i$. The first set of constraints ensure the degree-bounds are not violated more than $\alpha$. The second constraint ensures the weight is not violated by more than $\alpha$. The third set of constraints ensure the endpoints of every demand are connected.

\begin{align}
\text{minimize}           &\quad \alpha\enspace . \tag{{$\mathbb{PC\_IP}$}}\label{IP:good}\\
\forall v \in V &\quad\quad  \frac{1}{b_v}\sum_{i=1}^k\sum_{H \subseteq G}\deg_H(v)x_H^i \leq \alpha\label{cons:p1}\enspace .\\
&\quad\quad \frac{1}{w_{\text{opt}}}\sum_{i=1}^k\sum_{H \subseteq G}w(H)x_H^i \leq \alpha\label{cons:p2}\enspace .\\
\forall \sigma_i  &\quad\quad \sum_{H \in \mathcal{H}_i} {x_H^i} \geq 1\label{cons:c1}\enspace .\\\nonumber
\forall H \subseteq G, 1\leq i\leq k &\quad\quad \x_H^i \in \{0, 1\}\enspace .\\\nonumber
&\quad\quad \alpha > 0\enspace .
\end{align}

We are considering the online variant of the mixed packing and covering program. We are given the packing Constraints \eqref{cons:p1} and \eqref{cons:p2} in advance. At arrival of demand $\sigma_i$, the corresponding covering Constraint \eqref{cons:c1} is added to the program. We are looking for an online solution which is feasible at every online stage. Moreover the variables $x_H$ should be monotonic, i.e. once an algorithm sets $x_H=1$ for some $H$, the value of $x_H$ is 1 during the rest of the algorithm. Figure \eqref{fig:SSF} illustrates an example which indicates the difference between the solutions of \ref*{IP:good} and \ref{IP:sf}.
	\begin{figure} [!h]
		\centering
		\includegraphics[width=0.32\textwidth]{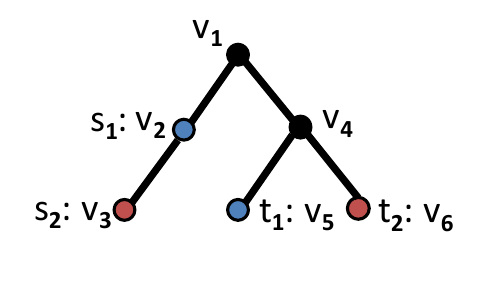}
		\caption{An example where every vertex has degree-bound 3 and every edge has weight 1. The first demand is $(v_2, v_5)$ and the second demand is $(v_3, v_6)$. The optimal solution for \ref{IP:sf} is a subgraph, say $H$, with the set of all edges and vertices, i.e. $H=G$. However an optimal solution for \ref{IP:good} is: Two subgraphs $H_1$ for the first request which has edges $\{e(v_1, v_2), e(v_1, v_4), e(v_4, v_5) \}$ and $H_2$ for the second request which has edges $\{e(v_2,  v_3), e(v_4, v_5), e(v_4, v_6)\}$. Note that $w(H)=5$ and $w(H_1)+w(H_2)=6$, since we have edge $e(v_4, v_5)$ in both $H_1$ and $H_2$. Moreover the number of edges incident with $v_4$ in the solution of \ref{IP:good} is 4, i.e. $\deg_{H_1}(v_4)+\deg_{H_2}(v_4)=4$.}
		\label{fig:SSF}
	\end{figure}

Let \pcopt{} and \lpopt{} denote the optimal solutions for \ref{IP:good} and \ref{IP:sf}, respectively. Lemma \ref{lm:sf1} shows given an online solution for \ref{IP:good} we can provide a feasible online solution for \ref{IP:sf} of cost \pcopt{}.

\begin{lemma}\label{lm:sf1}
	Given a feasible solution \{$\x$, $\alpha$\} for \ref{IP:good}, there exists a feasible solution \{$\x'$, $\alpha$\} for \ref{IP:sf}.
\end{lemma}
 
In the rest of this section, we show that we do not lose much by changing \ref{IP:sf} to \ref{IP:good}.\\
In particular Lemma \ref{mohem} shows $\pcopt \leq O(\log^2 n)\lpopt$. 
To this end, we first define the \textit{connective} list of subgraphs for a graph $G$, a forest $F$, and a list of demands $\sigma$. We then prove an existential lemma for such a list of subgraphs with a desirable property for any $\langle G, F, \sigma\rangle$. With that in hand, we prove $\pcopt \leq O(\log n)\lpopt$ in Lemma \ref{mohem}. 

Given a graph $G$, a list of demands $\sigma = \langle(s_1, t_1), (s_2, t_2) \ldots, (s_k, t_k)\rangle$, and a forest of $G$, $F$ we define a \textit{connective} list of subgraphs for $\langle G,F,\sigma\rangle$ in the following way:
\begin{definition}
Let $Q = \langle Q_1, Q_2, Q_3, \ldots,Q_k\rangle$ be a list of $k$ subgraphs of $F$. We say $Q$ is a connective list of subgraphs for $\langle G, F, \sigma\rangle$ iff for every  $1 \leq i \leq k$ there exists no cut disjoint from $Q_i$ that separates $s_i$ from $t_i$, but does not separate any $s_j$ from $t_j$ for $j < i$.
\end{definition}
The intuition behind the definition of connective subgraphs is the following: If $Q$ is a connective list of subgraphs for an instance $\langle G, F, \sigma\rangle$ then for every  $i$ we are guaranteed that the union of all subgraphs $\cup_{j=1}^i Q_i$ connects $s_i$ to $t_i$. In Lemma \ref{main} we show for every $\langle G, F, \sigma\rangle$, there exists a connective list of subgraphs for $\langle G, F, \sigma\rangle$, such that each edge of $F$ appears in at most $O(\log n)$ subgraphs of $Q$.
\begin{lemma}\label{main}
Let $G$ be a graph and $F$ be a forest in $G$. If $\sigma$ is a collection of $k$ demands $\langle(s_1, t_1), (s_2, t_2) \ldots, (s_k, t_k)\rangle$, then there exists a connective list of subgraphs $Q = \langle Q_1, Q_2, \ldots, Q_k\rangle$ for $\langle G,F,\sigma\rangle$ such that every edge of $F$ appears in at most $3\log |V(F)|$ $Q_i$'s.
\end{lemma}

Finally, we leverage Lemma \ref{main} to show $\pcopt \leq O(\log n)\lpopt$.
\begin{lemma}\label{mohem}
$\pcopt \leq O(\log n)\lpopt$.
\end{lemma} 


Finally, we leverage Lemma \ref{main} to show $\pcopt \leq O(\log^2 n)\lpopt$.
\begin{lemma}\label{mohem}
	$\pcopt \leq O(\log^2 n)\lpopt$.
\end{lemma}


This shows we can use \ref{IP:good} as an online mixed packing/covering IP to obtain an online solution for \oewdbsf{} losing a factor of $O(\log^2 n)$.\\
In Section \ref{sec:everything} we show this formulation is an online bounded frequency mixed packing/covering IP, thus we leverage our technique for such IPs to obtain a polylogarithmic-competitive algorithm for \oewdbsf{}.
\section{Online Bounded Frequency Mixed Packing/Covering IPs}\label{sec:OMPC}
	\newcommand{\sss}{\mathcal{S}}
	In this section we consider bounded frequency online mixed packing and covering integer programs. For every online mixed packing and covering IP with covering frequency $k$, we provide an online algorithm that violates each packing constraint by at most a factor of $O(k\log m)$, where $m$ is the number of packing constraints. We note that this bound is independent of the number of variables, the number of covering constraints, and the coefficients of the mixed packing and covering program. Moreover the algorithm is for integer programs, which implies obtaining an integer solution does not rely on (online) rounding.
	
	In particular we prove there exists an online $O(k\log m)$-competitive algorithm for any mixed packing and covering IP such that
every variable has covering frequency at most $k$,
	where the covering frequency of a variable $x_r$ is the number of covering constraints with a non-zero coefficient for $x_r$.
	
	 We assume that all variables are binary. One can see this is without loss of generality as long as we know every variable $x_r \in \{1, 2, 3, \ldots, 2^l\}$. Since we can replace $x_r$ by $l$ variables $y_r^1, \ldots, y_r^l$ denoting the digits of $x_r$ and adjust coefficients accordingly. Furthermore, for now we assume that the optimal solution for the given mixed packing and covering program is 1. In Theorem \ref{thm:radife} we prove that we can use a doubling technique to provide an $O(k\log m)$-competitive solution for online bounded frequency mixed packing and covering programs with any optimal solution. The algorithm is as follows. We maintain a family of subsets $\mathcal{S}$. Initially $\sss=\emptyset$. Let $\sss(j)$ denote $\sss$ at arrival of $C_{j+1}$. For each covering constraint $C_{j+1}$, we find a subset of variables $S_{j+1}$ and add $S_{j+1}$ to $\sss$. We find $S_{j+1}$ in the following way. For each set of variables $S$, we define a cost function $\tau_S(\sss(j))$ according to our current $\sss$ at arrival of $C_{j+1}$. We find a set $S_{j+1}$ that satisfies $C_{j+1}$ and minimizes $\tau_S(\sss(j))$. More precisely we say a set of variables $S$ satisfies $C_{j+1}$ if
	\begin{itemize}
		\item $\sum_{x_r \in S}{C_{j+1, r}x_r}\geq 1$, where $C_{j+1, r}$ denotes the coefficient of $C_{j+1}$ for $x_r$.
		\item For each packing constraint $P_i$, $\sum_{x_r \in S}{\frac{1}{k}P_{ir}} \leq 1$.
	\end{itemize}
	Now we add $S_{j+1}$ to $\sss$ and for every $x_r \in S_{j+1}$, we set $x_r=1$. We note that there always exists a set $S$ that satisfies $C_{j+1}$, since we assume there exists an optimal solution with value 1. Setting $S$ to be the set of all variables with value one in an optimal solution which have non-zero coefficient in $C_{j+1}$, satisfies $C_{j+1}$. It only remains to define $\tau_S(\sss(j))$. But before that we need to define $\Delta_i(S)$ and $F_i(\sss(j))$. For packing constraint $P_i$ and subset of variables $S$, we define $\Delta_i(S)$ as
	$\Delta_i(S):=\sum_{x_r\in S}{\frac{1}{k}P_{ir}}$.
	For packing constraint $P_i$ and $\sss(j)$, let
	\begin{equation}\label{eq:fdef}
	F_i(\sss(j)):=\sum_{S \in \sss(j)} \Delta_i(S)\enspace.
	\end{equation}
	Now let $\tau_S(\sss(j))=\sum_{i=1}^m\rho^{F_i(\sss(j))+\Delta_i(S)}-\rho^{F_i(\sss(j))}$,
	where $\rho>1$ is a constant to be defined later.
	
			\begin{algorithm}[!h]
				\textbf{Input:} Packing constraints $P$, and an online stream of covering constraints $C_1, C_2, \ldots$.
				
				\textbf{Output:} A feasible solution for online bounded frequency mixed packing/covering.
				
						\textbf{Offline Process:}
						\begin{algorithmic}[1]
							\STATE Initialize $\sss\gets \emptyset$.
						\end{algorithmic}
				
				\textbf{Online Scheme; assuming a covering constraint $C_{j+1}$ is arrived:}
				
				\begin{algorithmic}[1]
					\STATE $S_{j+1} \gets \argmin_{S }\{\tau_S(\sss(j))\ |\ S \text{ satisfies } C_{j+1}\}$.
					\FORALL {$x_r \in S_{j+1}$}
					\STATE $x_r \gets 1$.
					\ENDFOR
				\end{algorithmic}
				\caption{}
				\label{alg:OMPC}
			\end{algorithm}
	 Let $\x^*$ be an optimal solution, and $\x^*(j)$ denote its values at online stage $j$. We define $G_i(j)$ as
	 \begin{equation}\label{eq:gdef}
	 G_i(j):=\sum_{l=1}^j\sum_{r:C_{lr}>0} {\frac{1}{k}x^*_rP_{ir}}\enspace.
	 \end{equation}
	 Now we define a potential function $\Phi_j$ for online stage $j$.
	 \begin{equation}\label{eq:phi}
	 \Phi_j=\sum_{i=1}^{m} \rho^{F_i(\sss(j))} (\gamma - G_i(j))\enspace,
	 \end{equation}
	 where $\rho, \gamma > 1$ are constants to be defined later.
	 \begin{lemma}\label{lm:monotone}
	 	There exist constants $\rho$ and $\gamma$, such that $\Phi_j$ is non-increasing.
	 \end{lemma}
	 \begin{proof}
	 	We find $\rho$ and $\gamma$ such that $\Phi_{j+1}-\Phi_j \leq 0$. By the definition of $\Phi_j$, 	
	 	\begin{equation}\label{eq:pot}
	 	\Phi_{j+1}-\Phi_j=\sum_{i=1}^{m} \rho^{F_i(\sss(j+1))} (\gamma - G_i(j+1)) - \rho^{F_i(\sss(j))} (\gamma - G_i(j))\enspace.
	 	\end{equation}
	 	By Equation \eqref{eq:fdef}, $\rho^{F_i(\sss(j+1))}-\rho^{F_i(\sss(j))}=\rho^{F_i(\sss(j))+\Delta_i(S)} - \rho^{F_i(\sss(j))}$. Moreover by Equation \eqref{eq:gdef},
	 	$(\gamma - G_i(j+1))-(\gamma - G_i(j))=-\sum_{r:C_{j+1, r}>0} \frac{1}{k}x^*_rP_{ir}$.	 	
	 	For simplicity of notation we define $B_i(j+1):=\sum_{r:C_{j+1, r}>0} \frac{1}{k}x^*_rP_{ir}$.
	 	Thus we can write Equation \eqref{eq:pot} as:
	 	\begin{align}\label{eq:pot2}
	 	\Phi_{j+1}-\Phi_j=&\sum_{i=1}^{m} \rho^{F_i(\sss(j+1))} (\gamma - G_i(j)-B_i(j+1)) - \rho^{F_i(\sss(j))} (\gamma - G_i(j))
	 		 	\end{align}
	 		 	\vspace{-.7cm}
	 		 	\noindent
	 		 	\begin{align*}
	 	=&\sum_{i=1}^{m}{(\gamma-G_i(j))(\rho^{F_i(\sss(j))+\Delta_i(S)} - \rho^{F_i(\sss(j))})}-{\rho^{F_i(\sss(j+1))}B_i(j+1)} &\text{Since }G_i(j) \geq 0\\\nonumber
	 	\leq&\sum_{i=1}^{m}{\gamma(\rho^{F_i(\sss(j))+\Delta_i(S)} - \rho^{F_i(\sss(j))})}-{\rho^{F_i(\sss(j+1))}B_i(j+1)}
	 	&F_i(\sss(j+1)) \geq F_i(\sss(j))\\ \nonumber
	 	\leq&\sum_{i=1}^{m}{\gamma(\rho^{F_i(\sss(j))+\Delta_i(S)} - \rho^{F_i(\sss(j))})}-{\rho^{F_i(\sss(j))}B_i(j+1)}\enspace.
	 	\end{align*}
	 	Now according to the algorithm for each subset of variables $S'$ such that $\sum_{x_r \in S'}C_{j+1}(x_r)\geq 1$, either $\tau_{S}(\sss(j)) \leq \tau_{S'}(\sss(j))$ or there exists a packing constraint $P_i$ such that $\Delta_i(S') > 1$. In $B_i(j+1)$, we are considering variables $x_r$ such that $x^*_e=1$, thus for every $P_i$, $B_i(j+1) \leq 1$. Therefore setting $S'$ to be the set of variables $x_r$ such that $x^*_r=1$ and $C_{j+1, r}>0$, we have $\tau_{S}(\sss(j)) \leq \tau_{S'}(\sss(j))$. Thus $\sum_{i=1}^m\rho^{F_i(\sss(j))+\Delta_i(S)} - \rho^{F_i(\sss(j))} \leq \sum_{i=1}^m\rho^{F_i(\sss(j))+B_i(j+1)} - \rho^{F_i(\sss(j))}$. Therefore we can rewrite Inequality \eqref{eq:pot2} as
	 	\begin{align}
	 	\Phi_{j+1}-\Phi_j\leq&\sum_{i=1}^{m}{\gamma(\rho^{F_i(\sss(j))+B_i(j+1)} - \rho^{F_i(\sss(j))})}-{\rho^{F_i(\sss(j))}B_i(j+1)}
	 	\\\nonumber
	 	=&\sum_{i=1}^{m}{\rho^{F_i(\sss(j))}(\gamma \rho^{B_i(j+1)}-\gamma-B_i(j+1))}\enspace.
	 	\end{align}
	 	We would like to find $\rho$ and $\gamma$ such that $\Phi_j$ is non-increasing. We find $\rho$ and $\gamma$ such that for each packing constraint $P_i$, $\gamma \rho^{B_i(j+1)}-\gamma-B_i(j+1) \leq 0$. Thus
	 	\begin{align}
	 	\gamma\rho^{B_i(j+1)}-\gamma &\leq B_i(j+1) \quad\quad &\text{Since $0 \leq B_i(j+1) \leq 1$}\\
	 	\gamma\rho{B_i(j+1)}-\gamma &\leq B_i(j+1)\quad\quad &\text{By simplifying}\\
	 	\rho &\leq 1+1/\gamma\enspace.
	 	\end{align}
	 	Thus if we set $\rho \leq 1+1/\gamma$, $\Phi_j$ is non-increasing, as desired.
	 \end{proof}
	 
	 Now we prove Algorithm \ref{alg:OMPC} obtains a solution of at most $O(k\log m)$.
	 \begin{lemma}\label{lm:mohem}
		Given an online bounded frequency mixed packing covering IP with optimal value 1, there exists a deterministic integral algorithm with competitive ratio $O(k \log m)$, where $m$ is the number of packing constraints and $k$ is the covering frequency of the IP.
	 \end{lemma}

	 \begin{proof}
	 	By Lemma \ref{lm:monotone} for each stage $j$, $\Phi_{j+1} \leq \Phi_j$. Therefore $\Phi_j\leq \Phi_0=\gamma m$. Thus for each packing constraint $P_i$,
	 	\begin{equation}
	 	\rho^{F_i(\sss(j))}(\gamma-G_i(j)) \leq \gamma m\enspace.
	 	\end{equation}
	 	Thus,
	 	\begin{align}
	 	\rho^{F_i(\sss(j))} &\leq \frac{\gamma m}{(\gamma-G_i(j))}
	 	\leq \frac{\gamma m}{\gamma-1}\enspace. \quad\quad \text{Since $G_i(j) \leq 1$}
	 	\end{align} 
	 	Thus we can conclude 
	 	\begin{equation}\label{eq:last}
	 	F_i(\sss(j)) \in O(\log m)\enspace. 
	 	\end{equation}
	 	By definition of $F_i(\sss(j))$, $F_i(\sss(j))=\sum_{S\in \sss(j)}\Delta_i(S)=\sum_{S \in \sss(j)}\sum_{x_r \in S}\frac{1}{k}P_{ir}$. Since each variable $x_r$ is present in at most $k$ sets, $\frac{1}{k}P_i\cdot \x(j)\leq F_i(\sss(j))\enspace .$
	 	Thus by Inequality \eqref{eq:last} $P_i\x(j) \in O(k \log m)$, which completes the proof.
	 \end{proof}
	 
	 In Theorem \ref{thm:radife} we prove there exists an online $O(k \log m)$-competitive algorithm for bounded frequency online mixed packing and covering integer programs with any optimal value.
	 \begin{theorem}\label{thm:radife}
	 	Given an instance of the online mixed packing/covering IP, there exists a deterministic integral algorithm with competitive ratio $O(k \log m)$, where $m$ is the number of packing constraints and $k$ is the covering frequency of the IP.
	 \end{theorem}
\section{Putting Everything Together}\label{sec:everything}
In this section we consider the online mixed packing/covering formulation discussed in Section \ref{sec:rightLP} for \oewdbsf{} \ref{IP:good}. In this section we show this formulation is an online bounded frequency mixed packing/covering IP. Therefore we our techniques discussed in Section \ref{sec:OMPC} to obtain a polylogarithmic-competitive algorithm for \oewdbsf{}.

First we assume we are given the optimal weight $\wopt$ as well as degree bounds. We can obtain the following theorem.
\begin{theorem}
	Given the optimal weight $\wopt$, there exists an online deterministic algorithm which finds a subgraph with
	total weight at most $O(\log^3 n) w_{\opt}$ while the degree bound of a
	vertex is violated by at most a factor of $O(\log^3(n)$.
\end{theorem}
\begin{proof}
	By Lemma \ref{lm:sf1}, given a feasible online solution for \ref{IP:good} with violation $\alpha$, we can provide an online solution for \ref{IP:sf} with violation $\alpha$. Moreover by Lemma \ref{mohem}, $\pcopt \leq O(\log^2 n)\lpopt$. Thus given an online solution for \ref{IP:good} with competitive ratio $f$, there exists an $O(f \log n)$-competitive algorithm for \odbsf{}. We note that in \ref{IP:good} we know the packing constraints in advance. In addition every variable $x_H^i$ has non-zero coefficient only in the covering constraint corresponding to connectivity of the $i$-th demand endpoints, i.e. the covering frequency of every variable is 1. Therefore by Theorem \ref{thm:radife} there exists an online $O(\log m)$-competitive solution for \ref{IP:good}, where $m$ is the number of packing constraints, which is $n+1$. Thus there exists an online $O(\log^3 n)$-competitive algorithm for \odbsf{}. This means the violation for both degree bounds and weight is of $O(\log^3 n)$.
\end{proof}

Now we assume we are not given $\wopt$. The following theorems directly hold by applying the doubling technique mentioned in Section \ref{sec:doubling}.

\begin{theorem}
	There exists an online deterministic algorithm which finds a subgraph with
	total weight at most $O(\log^3 n) w_{\opt}$ while the degree bound of a
	vertex is violated by at most a factor of $O(\log^3(n)\log(\wopt))$.
\end{theorem}

Moreover if one favors the degree bound over total weight we  obtain the following bounds.
\begin{theorem}
	There exists an online deterministic algorithm which finds a subgraph with
	total weight at most $O(\log^3 n\frac{\log \wopt}{\log \log \wopt}) w_{\opt}$ while the degree bound of a
	vertex is violated by at most a factor of $O(\log^3 n\frac{\log \wopt}{\log \log \wopt})$.
\end{theorem}

\newpage
\bibliographystyle{abbrv}
\bibliography{dbbib}

\begin{thebibliography}{10}

\bibitem{AKR91a}
A.~Agrawal, P.~N. Klein, and R.~Ravi.
\newblock How tough is the minimum-degree steiner tree?: A new approximate
  min-max equality.
\newblock {\em Technical Report CS-91-49, Brown University}, 1991.

\bibitem{AAABN09}
N.~Alon, B.~Awerbuch, Y.~Azar, N.~Buchbinder, and J.~Naor.
\newblock The online set cover problem.
\newblock {\em SIAM Journal on Computing}, 39(2):361--370, 2009.

\bibitem{aspnes1997line}
J.~Aspnes, Y.~Azar, A.~Fiat, S.~Plotkin, and O.~Waarts.
\newblock On-line routing of virtual circuits with applications to load
  balancing and machine scheduling.
\newblock {\em Journal of the ACM (JACM)}, 44(3):486--504, 1997.

\bibitem{AAB96}
B.~Awerbuch, Y.~Azar, and Y.~Bartal.
\newblock On-line generalized steiner problem.
\newblock In {\em Proceedings of the seventh annual ACM-SIAM symposium on
  Discrete algorithms}, pages 68--74, 1996.

\bibitem{ABFP13}
Y.~Azar, U.~Bhaskar, L.~Fleischer, and D.~Panigrahi.
\newblock Online mixed packing and covering.
\newblock In {\em Proceedings of the Twenty-Fourth Annual ACM-SIAM Symposium on
  Discrete Algorithms}, pages 85--100. SIAM, 2013.

\bibitem{bauer1995degree}
F.~Bauer and A.~Varma.
\newblock Degree-constrained multicasting in point-to-point networks.
\newblock In {\em INFOCOM'95. Fourteenth Annual Joint Conference of the IEEE
  Computer and Communications Societies. Bringing Information to People.
  Proceedings. IEEE}, pages 369--376. IEEE, 1995.

\bibitem{BC97}
P.~Berman and C.~Coulston.
\newblock On-line algorithms for steiner tree problems.
\newblock In {\em Proceedings of the twenty-ninth annual ACM symposium on
  Theory of computing}, pages 344--353, 1997.

\bibitem{BN09}
N.~Buchbinder and J.~Naor.
\newblock The design of competitive online algorithms via a primal: dual
  approach.
\newblock {\em Foundations and Trends{\textregistered} in Theoretical Computer
  Science}, 3(2--3):93--263, 2009.

\bibitem{chakrabarty2015online}
D.~Chakrabarty, A.~Ene, R.~Krishnaswamy, and D.~Panigrahi.
\newblock Online buy-at-bulk network design.
\newblock In {\em FOCS}, 2015.

\bibitem{chaudhuri2006push}
K.~Chaudhuri, S.~Rao, S.~Riesenfeld, and K.~Talwar.
\newblock A push-relabel algorithm for approximating degree bounded msts.
\newblock In {\em Proceedings of the 33rd international conference on Automata,
  Languages and Programming-Volume Part I}, pages 191--201. Springer-Verlag,
  2006.

\bibitem{chaudhuri2009would}
K.~Chaudhuri, S.~Rao, S.~Riesenfeld, and K.~Talwar.
\newblock What would edmonds do? augmenting paths and witnesses for
  degree-bounded msts.
\newblock {\em Algorithmica}, 55(1):157--189, 2009.

\bibitem{Dehghani2016}
S.~Dehghani, S.~Ehsani, M.~Hajiaghayi, and V.~Liaghat.
\newblock Online degree-bounded steiner network design.
\newblock In {\em SODA}, 2016.

\bibitem{EV14}
A.~Ene and A.~Vakilian.
\newblock Improved approximation algorithms for degree-bounded network design
  problems with node connectivity requirements.
\newblock {\em STOC}, 2014.

\bibitem{ene2014improved}
A.~Ene and A.~Vakilian.
\newblock Improved approximation algorithms for degree-bounded network design
  problems with node connectivity requirements.
\newblock In {\em Proceedings of the 46th Annual ACM Symposium on Theory of
  Computing}, pages 754--763. ACM, 2014.

\bibitem{fukunaga2012iterative}
T.~Fukunaga and R.~Ravi.
\newblock Iterative rounding approximation algorithms for degree-bounded
  node-connectivity network design.
\newblock In {\em Foundations of Computer Science (FOCS), 2012 IEEE 53rd Annual
  Symposium on}, pages 263--272. IEEE, 2012.

\bibitem{FR90}
M.~F{\"u}rer and B.~Raghavachari.
\newblock An {NC} approximation algorithm for the minimum degree spanning tree
  problem.
\newblock In {\em Allerton Conf. on Communication, Control and Computing},
  pages 274--281, 1990.

\bibitem{FR94}
M.~F{\"u}rer and B.~Raghavachari.
\newblock Approximating the minimum-degree steiner tree to within one of
  optimal.
\newblock {\em Journal of Algorithms}, 17(3):409--423, 1994.

\bibitem{goemans2006minimum}
M.~X. Goemans.
\newblock Minimum bounded degree spanning trees.
\newblock In {\em Foundations of Computer Science, 2006. FOCS'06. 47th Annual
  IEEE Symposium on}, pages 273--282. IEEE, 2006.

\bibitem{G06}
M.~X. Goemans.
\newblock Minimum bounded degree spanning trees.
\newblock In {\em Foundations of Computer Science, 2006. FOCS'06. 47th Annual
  IEEE Symposium on}, pages 273--282, 2006.

\bibitem{HLP14}
M.~Hajiaghayi, V.~Liaghat, and D.~Panigrahi.
\newblock Near-optimal online algorithms for prize-collecting steiner problems.
\newblock In {\em Automata, Languages, and Programming}, pages 576--587. 2014.

\bibitem{HLP13}
M.~T. Hajiaghayi, V.~Liaghat, and D.~Panigrahi.
\newblock Online node-weighted steiner forest and extensions via disk
  paintings.
\newblock In {\em Foundations of Computer Science (FOCS), 2013 IEEE 54th Annual
  Symposium on}, pages 558--567, 2013.

\bibitem{IW91}
M.~Imase and B.~M. Waxman.
\newblock Dynamic {Steiner} tree problem.
\newblock {\em SIAM Journal on Discrete Mathematics}, 4(3):369--384, 1991.

\bibitem{KKN13}
R.~Khandekar, G.~Kortsarz, and Z.~Nutov.
\newblock On some network design problems with degree constraints.
\newblock {\em Journal of Computer and System Sciences}, 79(5):725--736, 2013.

\bibitem{klein2004approximation}
P.~N. Klein, R.~Krishnan, B.~Raghavachari, and R.~Ravi.
\newblock Approximation algorithms for finding low-degree subgraphs.
\newblock {\em Networks}, 44(3):203--215, 2004.

\bibitem{konemann2000matter}
J.~K{\"o}nemann and R.~Ravi.
\newblock A matter of degree: Improved approximation algorithms for
  degree-bounded minimum spanning trees.
\newblock In {\em Proceedings of the thirty-second annual ACM symposium on
  Theory of computing}, pages 537--546. ACM, 2000.

\bibitem{konemann2003primal}
J.~K{\"o}nemann and R.~Ravi.
\newblock Primal-dual meets local search: approximating mst's with nonuniform
  degree bounds.
\newblock In {\em Proceedings of the thirty-fifth annual ACM symposium on
  Theory of computing}, pages 389--395. ACM, 2003.

\bibitem{lau2009survivable}
L.~C. Lau, J.~Naor, M.~R. Salavatipour, and M.~Singh.
\newblock Survivable network design with degree or order constraints.
\newblock {\em SIAM Journal on Computing}, 39(3):1062--1087, 2009.

\bibitem{LS13}
L.~C. Lau and M.~Singh.
\newblock Additive approximation for bounded degree survivable network design.
\newblock {\em SIAM Journal on Computing}, 42(6):2217--2242, 2013.

\bibitem{MRSR98}
M.~V. Marathe, R.~Ravi, R.~Sundaram, S.~Ravi, D.~J. Rosenkrantz, and H.~B.
  Hunt~III.
\newblock Bicriteria network design problems.
\newblock {\em Journal of Algorithms}, 28(1):142--171, 1998.

\bibitem{GJ79}
R.~G. Michael and S.~J. David.
\newblock Computers and intractability: a guide to the theory of
  np-completeness.
\newblock {\em WH Freeman \& Co., San Francisco}, 1979.

\bibitem{NPS11}
J.~Naor, D.~Panigrahi, and M.~Singh.
\newblock Online node-weighted steiner tree and related problems.
\newblock In {\em Foundations of Computer Science (FOCS), 2011 IEEE 52nd Annual
  Symposium on}, pages 210--219, 2011.

\bibitem{N12}
Z.~Nutov.
\newblock Degree-constrained node-connectivity.
\newblock In {\em LATIN 2012: Theoretical Informatics}, pages 582--593. 2012.

\bibitem{oliveira2005survey}
C.~A. Oliveira and P.~M. Pardalos.
\newblock A survey of combinatorial optimization problems in multicast routing.
\newblock {\em Computers \& Operations Research}, 32(8):1953--1981, 2005.

\bibitem{QW11}
J.~Qian and D.~P. Williamson.
\newblock An o (logn)-competitive algorithm for online constrained forest
  problems.
\newblock In {\em Automata, Languages and Programming}, pages 37--48. 2011.

\bibitem{raghavachari1996algorithms}
B.~Raghavachari.
\newblock Algorithms for finding low degree structures.
\newblock In {\em Approximation algorithms for NP-hard problems}, pages
  266--295. PWS Publishing Co., 1996.

\bibitem{ravi2001approximation}
R.~Ravi, M.~V. Marathe, S.~Ravi, D.~J. Rosenkrantz, and H.~B. Hunt~III.
\newblock Approximation algorithms for degree-constrained minimum-cost
  network-design problems.
\newblock {\em Algorithmica}, 31(1):58--78, 2001.

\bibitem{ravi2006delegate}
R.~Ravi and M.~Singh.
\newblock Delegate and conquer: An lp-based approximation algorithm for minimum
  degree msts.
\newblock In {\em Automata, Languages and Programming}, pages 169--180.
  Springer, 2006.

\bibitem{LS07}
M.~Singh and L.~C. Lau.
\newblock Approximating minimum bounded degree spanning trees to within one of
  optimal.
\newblock In {\em Proceedings of the thirty-ninth annual ACM symposium on
  Theory of computing}, pages 661--670, 2007.

\bibitem{voss1992problems}
S.~Vo{\ss}.
\newblock Problems with generalized steiner problems.
\newblock {\em Algorithmica}, 7(1):333--335, 1992.

\end{thebibliography}
\newpage
\appendix
\section{Lower Bound}\label{sec:lower}
We present a (randomized) instance $I$ for mixed packing and covering linear
programs with the following parameters. $I$ consists of $m$ packing
constraints, $(2m-2)d$ variables and a variable is only contained in $\log_2d$
covering constraints. There exists an optimum integral solution for $I$ with
violation $\alpha=1$. Any (fractional) online algorithm $A$ incurs an expected
violation of at least $\Omega(\log_2m\cdot\log_2d)$, where the expectation is w.r.t.\
the randomized construction of $I$ and w.r.t.\ the random choices of $A$ in
case $A$ uses randomization. This gives Theorem~\ref{thm:mixedIPlower}.

%

\noindent
For the following description of the instance we assume that $m$ and $d$ are
powers of $2$. Consider a binary tree $T$ with $m$ leaf nodes. For each edge
$e$ in this tree we introduce $d$ variables, and we denote the set of variables
for edge $e$ with $X_e$. For each leaf node $v$ we introduce the packing
constraint
$\sum_{e\in P_v}\sum_{x\in X_e}x \le \alpha$,
where $P_v$ denotes the path from the root $r$ to $v$ in the tree. 

The covering constraints are constructed in an online manner according to a
random root-to-leaf path $P_\ell$ in the tree ($\ell$ is a random leaf node).
For a non-leaf node $v$ on this path (starting at the root and ending at the
parent of $\ell$) we construct $\log_2d$ covering constraint only involving
variables from $X_{e_L}\cup X_{e_R}$, where $e_L$ and $e_R$ denote the
child-edges of $v$. This sequence of covering constraints is constructed
according to the following lemma. For a set $X$ of variables we use $w(X)$ to
denote the total value assigned to these variables by the online algorithm.

\begin{lemma}
\label{lem:coveringconstraints}
Given two sets $X_L$ and $X_R$ of variables, each of cardinality $d$. There is
a randomized sequence of $\log_2d$ covering constraints over variables from $X_L\cup X_R$ 
such that any online algorithm has $E[w(X_L\cup X_R)]\ge \frac{1}{2}\log_2d$,
after fulfilling these constraints.

Furthermore, there exist variables $x_\ell\in X_L, x_r\in X_R$ such that
setting either of these variables to one already fulfills all constraints.
\end{lemma}
\begin{proof}
Define $A_L\subseteq X_L$ and $A_R\subseteq X_R$, as a set of \emph{active
  variables} of $X_L$ and $X_R$, respectively. Initially all variables in $X_L$
and $X_R$ are active, i.e., $A_L=X_L$ and $A_R=X_R$.

The constraints are constructed in $\log_2d$ rounds. In the beginning of the
$i$-th round ($i\in\{0,\dots,\log_2d-1\}$) $|A_L|=|A_R| = d/2^i$. We offer a
covering constraint on the current set of active variables, i.e., we offer
constraint
$\sum_{x\in A_L}x+\sum_{x\in A_R}x \ge 1$.
Then we remove 50\% of the elements from $A_L$, and 50\% of the elements from
$A_R$, at random, i.e., from each set we remove a random subset of cardinality
$d/2^{i+1}$.

After fulfilling the covering constraint for the $i$-th round, we have
$w(A_L)+w(A_R)\ge 1$. Removing random subsets from $A_L$ and $A_R$, removes 
variables of expected total weight at least $w(A_L)/2+w(A_R)/2\ge 1/2$. Hence,
the total expected weight removed from active sets during all rounds is at
least $1/2\cdot\log_2d$. This gives the bound on the expected weight of
variables.

Note that any variable that is active in the final round is contained in every
constraint; setting any of these variables to one fulfills all constraints.
Since, neither $A_L$ nor $A_R$ is empty in the final round the lemma follows.
\end{proof}

Note that by this construction an optimal offline algorithm can fulfill the
covering constraints of a node $v$ by setting a single variable from set $X_e$
to 1, where $e$ is the child-edge of $v$ that is \emph{not} on the path
$P_\ell$. Then, along any root-to-leaf path at most one set $X_e$ contains a variable
that is set to 1, and, hence, the maximum violation is $\alpha=1$. 

For the online algorithm we show that
$E[{\textstyle\sum_{e\in P_\ell}w(X_e)}]\ge \frac{1}{4}(\log_2m-1)\cdot\log_2d$,
which means that in expectation the maximum violation is at least $\Omega(\log
m\cdot\log d)$. To see this, we represent the path $P_\ell$ by a sequence of
left-right decisions. For $i=1,\dots,\log_2m-1$ define the binary random variable
$Y_i$ that is $1$ if at the $i$-th node the path $P_\ell$ continues left, and
is $0$ otherwise.
Further, we use $W_i^L$ and $W_i^R$ to denote the random variable that describes
the total weight assigned to variables in $X_{e_L}$ and $X_{e_R}$,
respectively, where $e_L$ and $e_R$, are the child-edges of the $i$-th node on
path $P_\ell$. Then,
\begin{equation*}
\begin{split}
E\Bigg[{\sum_{e\in P_\ell}w(X_e)}\Bigg]
&=E\Bigg[\sum_{i=1}^{\log_2m-1}Y_i\cdot W_i^L+(1-Y_i)\cdot W_i^R\Bigg]
=\sum_{i=1}^{\log_2m-1}E[Y_i]\cdot E[W_i^L]+E[1-Y_i]\cdot E[W_i^R]\\
&=\sum_{i=1}^{\log_2m-1}\frac{1}{2}E[W_i^L+W_i^R]
\ge \frac{1}{4}(\log_2m-1)\log_2d\enspace.
\end{split}
\end{equation*}
Here, the second equality follows as $Y_i$ is independent from $W_i^L$ and
$W_i^R$; the last inequality holds due to Lemma~\ref{lem:coveringconstraints}.


\section{Omitted Proofs of Section \ref{sec:rightLP}}

\begin{proofof}{Lemma \ref{lm:sf1}}
Let $(\alpha,x)$ be an optimal solution of \ref{IP:sf} for a given graph $G$ and a sequence of demands $\sigma$. We define subgraph $F$ as the union of all edges of $G$ whose value in $x$ is equal to 1. Since constraints of type \eqref{IP:sfc1} ensure that the demands are satisfied, every pair $(s_i,t_i)$ is connected in $F$. Without loss of generality we assume $F$ is a forest, since otherwise removing any edge from a cycle of $F$ provides a better solution for \ref{IP:sf} which contradicts the optimality of $(\alpha,x)$. Now, according to Lemma \ref{main}, there exists a connective list of subgraphs $Q$ for $\langle G, F, \sigma\rangle$ such that every edge of $F$ appears in at most $3\log^2 n$ subgraphs of $Q$ where $n$ is the size of the graph. We construct a solution $(\alpha',x')$ to \ref{IP:good} in the following way:

 Since $Q$ is a connective list of subgraphs for $\langle G, F, \sigma\rangle$, every demand $(s_i,t_i)$ is connected in $\cup_{j=1}^i Q_j$. Therefore, for every $i$, $Q_i \in \mathcal{H}_i$ holds. We set $\alpha' = (3\log^2 n) \alpha$, $x'^i_{Q_i} = 1$ for every $1 \leq i \leq k$, and $x'^i_H = 0$ for all other variables. Note that since $Q_i \in \mathcal{H}_i$ for all $i$, then constraints of type \eqref{cons:c1} are trivially satisfied. We define $\hat{x}_e$ for an edge $e \in E(G)$ as $\hat{x}_e = \sum_{i=1}^k \sum_{H \ni e} x'^i_H$. Since every edge of $F$ appears in at most $3 \log^2 n$ subgraphs of $Q$ and $x$ is a feasible solution for \ref{IP:sf}, for every vertex $v$ in $V(G)$ we have
\begin{eqnarray*}
\frac{1}{b_v} \sum_{i=1}^k \sum_{H \subseteq G} \mathsf{deg}_H(v)x'^i_H & = & \frac{1}{b_v} \sum_{e \in \delta(v)} \hat{x}_e\\
& \leq & (3 \log^2 n) \frac{1}{b_v} \sum_{e \in \delta(v)} x_e \\
& \leq & (3 \log^2 n) \alpha\\
& = & \alpha'
\end{eqnarray*} 
and thus all constraints of type \eqref{cons:p1} are satisfied by $x'$. Moreover, 
\begin{eqnarray*}
	\frac{1}{w(opt)} \sum_{i=1}^k \sum_{H \subseteq G} w(H)x'^i_H & = & \frac{1}{w(opt)} \sum_{e \in E(G)} w(e)\hat{x}_e\\
	& \leq & (3 \log^2 n) \frac{1}{w(opt)} \sum_{e \in E(G)} w(e) \\
	& \leq & (3 \log^2 n) \alpha\\
	& = & \alpha'
\end{eqnarray*} 
and hence $x'$ meets Constraint \eqref{cons:p2}. Thus, $x'$ is a feasible solution for \ref{IP:good} and the proof is complete.
\end{proofof}
\section{Omitted Proofs of Section \ref{sec:OMPC}}
\begin{proofof}{Theorem \ref{thm:radife}}
		 	By Lemma \ref{lm:mohem}, if the optimal solution of the IP is 1, then there exists a deterministic algorithm with competitive ratio $O(k \log m)$. We show one can use a doubling technique to obtain an $O(k \log m)$-competitive algorithm without such assumption.
	 	
	 	We start by guessing $\alpha=1$. We use Algorithm \ref{alg:OMPC} to find an online solution. Since Algorithm \ref{alg:OMPC} is $O(k \log m)$-competitive if $\alpha =1$, there exists a constant $\beta$ such that no packing constraint is violated by more than $\beta(k\log m)$. At each online stage $j$ if updating the solution violates a packing constraint by more than $\beta(k\log m)$, then $\alpha>1$. Thus we say we go to the next phase and do the following. First we divide every coefficient in packing constraints by 2. Note that this is equivalent to doubling $\alpha$. Then we remove all previous solutions of the algorithm (we basically ignore them). We use Algorithm \ref{alg:OMPC} again on the new IP from the beginning. We claim that we lose a factor of at most 4 due to these updates.
	 	
	 	At each phase every packing constraint is violated by at most $\beta k \log m$. The total violation of each packing constraint after $l$ phases is
	 	\begin{equation}
	 	\sum_{i=1}^{l}{\beta 2^{i-1}k\log m}\leq \beta 2^ik\log m\enspace.
	 	\end{equation} 
	 	Moreover our guess about $\alpha$ is not exact, since we know $2^{i-1}< \alpha \leq 2^i$. Thus we might violate each packing constraint by an additional factor of $2$. Thus we can use Algorithm \ref{alg:OMPC} to obtain a deterministic integral algorithm with competitive ratio $4\beta(k \log m)$.
\end{proofof}

\section{Doubling Method}\label{sec:doubling}
In this section we discuss how we can address the issue that we might not know $w_{\opt}$ in advance. In particular by Lemma \ref{lm:mohem} if we are given the degree bounds and $w_{\text{opt}}$, the optimal solution for \ref{IP:good} is 1. Thus Algorithm \ref{alg:OMPC} provides an online solution for \ref{IP:good} of $O(\log m)$, where $m$ is the number of packing constraints. We show we can start with an initial value for $\wopt$ and update it through several phases. We use Algorithm \ref{alg:OMPC} at each phase. Eventually we lose a factor because of the number of phases we run Algorithm \ref{alg:OMPC} and a factor of our approximation for $\wopt$.

More precisely, let $w'_{\text{opt}}$ be our guess for $\wopt$. Initially we set $w'_{\text{opt}}=1$. We write \ref{IP:good} according to $w'_{\text{opt}}$ and use Algorithm \ref{IP:good} to obtain online solution $\x$. By Lemma \ref{lm:mohem} if $w'_{\text{opt}} \leq \wopt$, there exists a constant $c$ such that every packing constraint is violated by a factor of at most $c\log m$ in solution $\x$. At any online step, if a packing constraint is going to be violated by more that $c \log m$, we stop the algorithm. Now we know $w'_{\text{opt}} > \wopt$, thus we update $w'_{\text{opt}}=w'_{\text{opt}} \times r$. We ignore the current solution $\x$. We write \ref{IP:good} by the updated value of $w'_{\text{opt}}$ and do the same until no packing constraint is violated by more than $c \log m$.

Now we analyze the competitive ratio of the current solution. Let $l$ denote the number of phases we updated $w'_{\text{opt}}$ and run Algorithm \ref{alg:OMPC}.

\begin{lemma}\label{lm:update1}
If $r \geq 2$ each degree constraint is violated by $O(l\log m)$, and the weight constraint is violated by $O(r\log m)$, where $m$ is the number of packing constraints.
\end{lemma}
\begin{proof}
Since at each phase no packing constraint is violated by more than $c \log m$, the total violation for every packing constraint is no more than $lc \log m$. Thus the total violation in every degree bound is $O(l\log m)$. 

At each phase, the packing constraint corresponding to weights is violated by $c \log m$, however $w'_{\text{opt}}$ is changing as well. The total violation of the weights is $\sum_{i=0}^{l-1}{r^i c\log m}$, since at each phase $i$, $w'_{\text{opt}}=r^{i-1}$. Since $r\geq 2$, $\sum_{i=0}^{l-1}{r^i c\log m} < 2r^{l-1}=2w'_{\text{opt}}$. Therefore the packing constraint corresponding to weights is also violated by no more than $2c \log m$. However $\frac{1}{r}w'_{\text{opt}} < \wopt \leq w'_{\text{opt}}$. Thus the total weight is no more than $2rc \log m\wopt$.
\end{proof}

The following two lemmas follow directly from Lemma \ref{lm:update1}.
\begin{lemma}\label{lm:update2}
Given graph $G$ and degree bounds. let $\wopt$ denote the optimal weight for \oewdbsf{}. There exists an online algorithm for \ref{IP:good} that violates the weight packing constraint by $O(\log n)$ and violates each degree bound packing constraints by $O(\log n \log \wopt)$.
\end{lemma}
\begin{proof}
	Let $r=2$. The number of phases $l=O(\log \wopt)$. Thus by Lemma \ref{lm:update1} there exists an online algorithm for \ref{IP:good} that violates the weight packing constraint by $O(\log n)$ and violates each degree bound packing constraints by $O(\log n \log \wopt)$.
\end{proof}
\begin{lemma}\label{lm:update3}
Given graph $G$ and degree bounds. let $\wopt$ denote the optimal weight for \oewdbsf{}. There exists an online algorithm for \ref{IP:good} that violates the weight packing constraint and each degree bound packing constraints by $O(\log n \frac{\log \wopt}{\log \log \wopt})$.
\end{lemma}
\begin{proof}
		Let $r=\frac{\log \wopt}{\log \log \wopt}$. The number of phases $l=O(\frac{\log \wopt}{\log \log \wopt})$. Thus by Lemma \ref{lm:update1} there exists an online algorithm for \ref{IP:good} that violates the weight packing constraint and each degree bound packing constraints by $O(\log n \frac{\log \wopt}{\log \log \wopt})$.
\end{proof}

\section{A Polynomial Algorithm for Finding the Path with the Minimum Cost}
Here we present a polynomial time method to find a set $S$ that satisfies $j+1$-th constraint and minimizes $\tau_S(\mathcal{S}(j))$. In our problem we can assume every two endpoints of previous demands have been contracted into one node. With this assumption the set $S$ is in form of a path from $s_j$ to $t_j$. The method we use is dynamic programming. Let $dp$ be a dynamic table such that $dp(l,v,w)$ denotes the minimum degree cost of a path with length $l<n$ and weight $w$ from $s_j$ to vertex $v\in V$. The degree cost of a path $P$ is the sum of the terms corresponding to degree constraints in the calculation of $\tau_P(\mathcal{S}(j))$. The set $S$ we are looking for has some length $l'$ and some weight $w'$. Knowing these two parameters, we can calculate $\tau_S(\mathcal{S}(j))$ using $w'$ and $dp(l',t_j,w')$.

If the maximum weight of an edge be a polynomial of $n$ then the number of entries in the dynamic table $dp$ is polynomial. Now we briefly explain how to calculate the value of each $dp(l,v,w)$. Initially set $dp(0, s_j, 0)$ to zero and $dp(l,v,w)$ to infinity for other entries. At every moment each entry with cost less than infinity represents a path from $s_j$. At every iteration we update new entries by adding one vertex to the existing paths.

We can now find the path minimizing $\tau_S(\mathcal{S}(j))$ by looking over $dp(l', t_j, w')$, for all $0\leq l'\leq n$ and $0\leq w' \leq n\max_{e\in E}\{w(e)\}$. Since the size of $dp$ and the number of iterations are polynomials of $n$, the whole method runs in polynomial time.

\section{Rounding Lemma}\label{sec:rounding}
For a given tree $T=(E(T), V(T))$, let $\mathcal{\pi}=\langle \pi_1, \ldots, \pi_n \rangle$ be a permutation on $V(T)$, and $\mathcal{F}$ be a collection of subsets of $E(T)$. Suppose a probability $p_{\pi_i,\pi_j}$ is assigned to every $\mathcal{F}_{\pi_i,\pi_j}$ in a way that $\sum_{j=1}^{i-1}p_{\pi_i,\pi_j}=1$, for every $1<i\leq n$. We define the load of $p$ on an edge $e$ as $L_p(e):=\sum_{(i, j):e\in \mathcal{F}_{\pi_i, \pi_j}}p_{\pi_i, \pi_j}$ and the overall load $L_p$ as $L_p:=\max_{e\in E(t)}\{L_p(e)\}$.

The following lemma states the existence of a rounding $q$ for $p$, which with a high probability has a bounded load with respect to $L_p$ and $\log n$.

\begin{lemma} \label{rounding_lemma}
	There exists a new assignment $q_{\pi_i, \pi_j}\in\{0,1\}$ to every $\mathcal{F}_{\pi_i, \pi_j}$, such that $\sum_{j=1}^{i-1}q_{\pi_i,\pi_j}=1$, for every $1<i\leq n$, and the new load $L_q$ is at most $O(\max\{L_p, \log n\})$ with a high probability.
\end{lemma}

\begin{proof}
	For every $1<i\leq n$ we need to set at least one of $q_{\pi_i, \pi_1},\ldots,q_{\pi_i, \pi_{i-1}}$ to 1.
	To do so, we select one of the $q_{\pi_i, \pi_j}$'s using the following random process. Take a random number $r\in [0,1]$. Let $j$ be the smallest index for which $\sum_{k=1}^{j}p_{\pi_i, \pi_k}\geq r$. Set $q_{\pi_i, \pi_j}$ to 1 and the remaining to 0. Note that the probability of every $q_{\pi_i, \pi_j}$ being 1 is exactly $p_{\pi_i, \pi_j}$.
	
	Using this random process, the expected load of $q$ on every edge $e$ remains the same as the load of $p$ on that edge, because
	\begin{align*}
		E[L_q(e)]&=E\bigg[\sum_{(i,j):e\in\mathcal{F}_{\pi_i, \pi_j}}{q_{\pi_i, \pi_j}}\bigg]\\&=\sum_{(i,j):e\in\mathcal{F}_{\pi_i, \pi_j}}{E[q_{\pi_i, \pi_j}]}\\&=\sum_{(i,j):e\in\mathcal{F}_{\pi_i, \pi_j}}{Pr[q_{\pi_i, \pi_j}=1]}\\&=\sum_{(i,j):e\in\mathcal{F}_{\pi_i, \pi_j}}{p_{\pi_i, \pi_j}}=L_p(e) \enspace .	
	\end{align*}
	
	Moreover, $L_q(e)$ is in fact the summation of a number of binary random variables which are not positively correlated \footnote{In particular, $q_{\pi_{i_1}, \pi_{j_1}}$ and $q_{\pi_{i_2}, \pi_{j_2}}$ are independent for $i_1\neq i_2$, and are negatively correlated for $i_1=i_2$ and $j_1\neq j_2$.}. Therefore, this summation can be upper bounded by the Chernoff bound: $$Pr[L_q(e)>(1+\delta)L_p(e)] \leq e^{-\frac{\delta L_p(e)}{3}} \enspace .$$
	
	In the above variation of the Chernoff bound, $\delta>1$. Mention that to achieve a small enough probability \footnote{At most $O(\frac{1}{n^2})$.}, it suffices for $\delta$ to be at least $\Omega(\log n/L_p(e))$. Finally, we use the Union bound to show that $L_q\in O(\max\{L_p, \log n\})$ with a high probability.
	
	\begin{align*}
		Pr[L_q\in O(\max\{L_p, \log n\})] &= Pr[\forall e\in E(T): L_q(e)\in O(\max\{L_p(e), \log n\})] \\
		&= 1 - Pr[\exists e\in E(T): L_q(e)\notin O(\max\{L_p(e), \log n\})]\\
		&\geq 1 - \sum_{e\in E(T)} Pr[L_q(e)\notin O(\max\{L_p(e), \log n\})]\\
		&\geq 1 - \sum_{e\in E(T)} Pr[L_q(e)>c\big(1+\frac{\log n}{L_p(e)}\big)L_p(e)] \\ &\geq 1-\sum_{e\in E(T)}O(\frac{1}{n^2}) = 1-O(\frac{1}{n})\enspace .
	\end{align*}
	
\end{proof}

\section{Reduction from weight guarantee to edge-wise guarantee}\label{sec:redavglemma}
\begin{lemma}\label{mohem}
	Let $\mathcal{R}$ be a subspace of $\mathbb{R}^n$ that contains all points of $\mathbb{R}^n$ with non-negative coordinates and $P$ be a convex set of points in $\mathcal{R}$. If for every point $\hat{x} = (x_1,x_2,\ldots,x_n) \in \mathcal{R}$ there exists a point $\hat{p} \in P$ such that
	$$\hat{p}.\hat{x} \leq k\sum_{i=1}^n x_i$$
	then $P$ contains a point $\hat{r}$ such that $\max_{i=1}^n r_i \leq k$.
\end{lemma}
\begin{proof}
	We define $P'$ as the set of all points in $\mathbb{R}^n$ whose all indices are greater than or equal to the corresponding indices of a point in $P$. In other words 
	$$P' = \{\hat{p} \in \mathbb{R}^n | \exists \hat{q} \in P \text{ such that } p_i \geq q_i \text{ for all } 1 \leq i \leq n\}.$$
	We show in the rest that $(k,k,\ldots,k) \in P'$ which immediately implies the lemma. To this end, suppose for the sake of contradiction that $(k,k,\ldots,k) \notin P'$. Note that, since $P$ is a convex set, so is $P'$. Therefore, there exists a hyperplane that separates all points of $P'$ from point $(k,k,\ldots,k)$. More precisely, there exist coefficients $h_0,h_1,\ldots,h_n$ such that 
	\begin{equation} \label{firstin}
		\sum_{i = 1}^n h_ip_i > h_0
	\end{equation}
	for all points $\hat{p} \in P'$ and 
	\begin{equation}\label{secondin}
		\sum_{i = 1}^n h_ik < h_0.
	\end{equation}
	Due to the construction of $P'$ we are guaranteed that all coefficients $h_1, h_2, \ldots, h_n$ are non-negative numbers since otherwise for any index $i$ such that $h_i < 0$ there exists a point in $P'$ whose $i$'th index is large enough to violate Inequality \eqref{firstin}. Now let $\hat{x} = (h_1, h_2, \ldots, h_n)$. By Inequalities \eqref{firstin} and \eqref{secondin} we have 
	$$\hat{p}.\hat{x} = \sum_{i = 1}^n x_ip_i = \sum_{i = 1}^n h_ip_i > h_0 \geq k \sum_{i=1}^n h_i = k\sum_{i=1}^n x_i$$
	for every $\hat{p} \in P'$ which means there is no $\hat{p} \in P$ such that $\hat{p}.\hat{x} \leq k \sum_{i=1}^n x_i$. This contradicts the assumption of the lemma.
\end{proof}

\end{document}